\documentclass[12pt]{article}
\usepackage{lMac}
\usepackage{becMac}

\newcommand{\An}{\mathrm{An}}
\newcommand{\lb}{\mathrm{lb}}
\newtheorem{mcorollary}[theorem]{``Corollary''}
\newcommand{\Bot}{S_{\mathrm{Bot}}}

\def\showintremarks{n}   
\def\showissues{n}       
\def\SHOWissues{n}       

\begin{document}
\title{The Small Field Parabolic Flow for Bosonic Many--body Models:\\
        \Large Part 3 --- Nonperturbatively Small Errors}

\author{Tadeusz Balaban}
\affil{\small Department of Mathematics \authorcr
       Rutgers, The State University of New Jersey \authorcr
       tbalaban@math.rutgers.edu\authorcr
       \  }

\author{Joel Feldman\thanks{Research supported in part by the Natural 
                Sciences and Engineering Research Council 
                of Canada and the Forschungsinstitut f\"ur 
                Mathematik, ETH Z\"urich.}}
\affil{Department of Mathematics \authorcr
       University of British Columbia \authorcr
       feldman@math.ubc.ca \authorcr
       http:/\hskip-3pt/www.math.ubc.ca/\squig feldman/\authorcr
       \  }

\author{Horst Kn\"orrer}
\author{Eugene Trubowitz}
\affil{Mathematik \authorcr
       ETH-Z\"urich \authorcr
       knoerrer@math.ethz.ch, trub@math.ethz.ch \authorcr
       http:/\hskip-3pt/www.math.ethz.ch/\squig knoerrer/}


\maketitle

\begin{abstract}
\noindent
This paper is a contribution to a program to see symmetry breaking in a
weakly interacting many Boson system on a three dimensional lattice at 
low temperature.  It is part of an analysis of the  ``parabolic flow'' which 
exhibits the formation of a ``Mexican hat'' potential well. 
Here we provide arguments that suggest,
but do not completey prove, that the difference between the ``small field'' approximation, 
analyzed in \cite{PAR1,PAR2}, and full model is nonperturbatively small.

\end{abstract}


\newpage

As part of the program to see symmetry breaking in an interacting many 
Boson system on a three dimensional lattice in the thermodynamic limit,
we analyze in \cite{PAR1,PAR2} the ``small field'' approximation 
to the ``parabolic flow'' which exhibits the formation of a potential 
well.  In this paper,
we argue that the errors made with this approximation are nonperturbatively
small, that is they are  smaller than any power of the coupling constant.
This note does not provide a proof of this fact; however
we feel that the arguments given here
can provide the core of such a proof.

The \emph{first simplification} leading to the  ``small field'' 
approximation is a simplification of the starting point.
The outcome of the previous flow \cite{UV} (which treats the 
temporal ultraviolet problem in imaginary time) represents the 
partition function as a sum over ``large--field/ small--field'' 
decompositions of space. All but one term in this sum are 
nonperturbatively small, and the first simplification is to continue 
the flow with only this one term. It is of the form
\begin{equation}\label{eqnSFCstarting point}
\cZ_\In^{|\cX_0|}\int \Big[ \smprod_{x\in \cX_0} 
        \sfrac{ d\psi(x)^\ast\wedge d\psi(x)}{2\pi \imath}\Big] \,
  e^{\cA_0(\psi^*,\psi) }\chi_0(\psi)
\end{equation}
where $\cZ_\In$ is a normalization factor, $\cX_0$ is a unit lattice,
the action $\cA_0$ has a very specific form and $\chi_0$ is a function
with compact support that implements the small field cutoff. See 
\cite[(\eqnHTstartingpoint) and (\eqnHTaIn)]{PAR1}.

The output of the $n^{\rm th}$ renormalization group step is 
an approximation to the partition function that is  a constant
times a functional integral over a space of complex valued 
fields $\psi$ on a unit lattice $\cX_0^{(n)}$.
One can write the action in this functional integral as function $\fA_n(\psi^*,\psi)$
of the field $\psi$ and its complex conjugate $\psi^*$, where $\fA_n(\psi_*,\psi)$
is an analytic function\footnote{The exact form of this function is stated
in the main theorem  \cite[Theorem \thmTHmaintheorem]{PAR1}} of two independent complex fields $\psi_*,\,\psi$.
The domain of integration at this point is a bounded subset $\Int(n)$
of the space of complex valued fields on  $\cX_0^{(n)}$.
See \eqref{eqnALFIntncDef} below.
A block spin transformation amounts to rewriting the functional
integral as
\begin{equation}\label{eqnSFCintrofunctint}
 \sfrac{1}{N^{(n)}}\hskip-2pt\int\hskip-2pt\Big[\hskip-5pt
 \prod_{y\in\cX_{-1}^{(n+1)}}\hskip-3pt
  \sfrac{d\th(y)^*\wedge d\th(y)}{2\pi i}\Big]
\int_{\Int(n)} \Big[\prod_{x\in\cX_0^{(n)}} 
                   \sfrac{d\psi(x)^*\wedge d\psi(x)}{2\pi i}\Big]\ 
  e^{-a L^{-2}\|\th-Q\psi\|^2_{-1}
    - \fA_n(\psi^*,\psi)}
\end{equation}
where $N^{(n)}$ is the normalization constant for the Gaussian integral
over $\th$,  $\cX_{-1}^{(n+1)}$ is a sublattice of $\cX_0^{(n)}$   and
$Q$ is an averaging operator defined in 
\cite[Definition \defHTblockspintr]{PAR1}.
The goal is to perform, for any fixed $\theta$,  the $\psi$ integral in 
\eqref{eqnSFCintrofunctint} to obtain a functional integral
representation of the partition function in the $\theta$ variables.
We view this $\psi$ integral as an integral of a holomorphic differential
form in the $2 \big|\cX_0^{(n)}\big|$ complex variables
$\psi_*(x), \psi(x),\,x\in \cX_0^{(n)}\,$ over the set 
$\,  D = \set{ (\psi_*, \psi)\in \Int(n) \times \Int(n) }{ \psi_* =\psi^* }\,$:
\begin{equation}\label{eqnSFCintrointernalintegral}
\int_{D} \Big[\prod_{x\in\cX_0^{(n)}} 
                   \sfrac{d\psi(x)_*\wedge d\psi(x)}{2\pi i}\Big]\ 
  e^{-a L^{-2}\<\th^*-Q\psi_* \,,\,\th-Q\psi\>_{-1}
    -\fA_n(\psi_*,\psi)}
\end{equation}
See \cite{ParOv}, Step 3. Observe that $D$ has $2 \big|\cX_0^{(n)}\big|$ real dimensions.

We will use stationary phase to evaluate the integral 
\eqref{eqnSFCintrointernalintegral}.  To do so, we want to determine,
for each fixed value of $\theta$, an approximate critical point 
$\,\big( \psi_{*n}(\th^*\!,\th) ,\,\psi_{n}(\th^*\!,\th) \big)\,$
for the map
\begin{equation}\label{eqnSFCintrocriticalpoint}
(\psi_*,\psi) \mapsto -a L^{-2}\<\th^*-Q\psi_* \,,\,\th-Q\psi\>_{-1}
    -\fA_n(\psi_*,\psi)
\end{equation}
This approximate critical point lies $\Int(n)\times \Int(n)$  only 
if $\theta$ is not too big. Below, we argue that for large $\theta$, 
the $\psi$ integral in \eqref{eqnSFCintrointernalintegral}
gives nonperturbatively small contributions. 
Therefore we make an {\it approximation} by restricting the variable 
$\theta$ in \eqref{eqnSFCintrofunctint} to a bounded
subset $\check\Int(n)$.

As pointed out in \cite[Step 3]{ParOv}, for general 
$\theta\in \check \Int(n)$, the critical point of \eqref{eqnSFCintrocriticalpoint}
does not fulfil the reality condition $\psi_{*n}(\th^*,\th) =\psi_{n}(\th^*,\th)^* $.
In particular it does not lie in the domain of integration $D$.
We choose a bounded subset\footnote{We wish to integrate over
a neighborhood of the critical point. So we make a change of
variables to ``fluctuation fields $\de\psi_*=\psi^*-\psi_{*n}(\th^*,\th)$,
$\de\psi=\psi-\psi_n(\th^*,\th)$''. The condition in the set below
is a reality condition on the fluctuation fields.} 
$S$ of
$$
\set{(\psi_*,\psi)}{ \psi_* - \psi_{*n}(\th^*,\th)  
               = \big(\psi - \psi_{n}(\th^*,\th) \big)^*}
$$
containing $\,\big( \psi_{*n}(\th^*,\th) ,\,\psi_{n}(\th^*,\th) \big)\,$,
and a $2 \big|\cX_0^{(n)}\big|+1$  dimensional set $\cY$ whose boundary consists
of $D$, $S$ and some other component. Below we argue that the integral
of 
\begin{equation*}
 \Big[\prod_{x\in\cX_0^{(n)}} 
                   \sfrac{d\psi(x)_*\wedge d\psi(x)}{2\pi i}\Big]\ 
  e^{-a L^{-2}\<\th^*-Q\psi_* \,,\,\th-Q\psi\>_{-1}
    -\fA_n(\psi_*,\psi)}
\end{equation*}
over $\partial \cY \setminus (S\cup D)$ is nonperturbatively small.
This, combined with Stokes' theorem, would justify our 
\emph{last approximation}, which is the replacement
of \eqref{eqnSFCintrointernalintegral} with
\begin{equation*}
\int_{S} \Big[\prod_{x\in\cX_0^{(n)}} 
                   \sfrac{d\psi(x)_*\wedge d\psi(x)}{2\pi i}\Big]\ 
  e^{-a L^{-2}\<\th^*-Q\psi_* \,,\,\th-Q\psi\>_{-1}
    -\fA_n(\psi_*,\psi)}
\end{equation*}

An important ingredient in the argument that the above approximations 
are justified is that, at the points considered, the 
effective action
\begin{equation*}
a L^{-2}\<\th^*-Q\psi_* \,,\,\th-Q\psi\>_{-1}
    + \fA_n(\psi_*,\psi)
\end{equation*}
has a large, positive, real part. Though positivity is suggested by 
the quadratic and quartic terms in the explicit form of the action 
(see \cite[Definition \defHTblockspintr]{PAR1}), we have
to pay close attention since the fields $\psi_*,\psi$ are complex valued.

This note can be considered as a complement to \cite{PAR1,PAR2} and 
uses the notation introduced there. This notation is summarized in 
\cite[Appendix \appDefinitions]{PAR1}.

We emphasise again that this note is intended to provide motivation 
rather than a proof. Some of the bounds are not uniform in 
the volume $\cX_0$.  Furthermore some of the statements we make are 
handwavy. We concentrate on showing where the nonperturbatively 
small factors come from. A rigorous construction, with bounds 
uniform in the volume, would entail expressing the errors as sums 
over ``large field subsets'' $\fL\subset\cX_0^{(n)}$ and exhibiting 
bounds which include a nonperturbatively small factor for each point 
of each $\fL$, as was done in \cite{UV}.

\bigskip
As said above, we start with the approximation \eqref{eqnSFCstarting point}
for the partition function $\Tr\, e^{-\sfrac{1}{kT}\,(H-\mu N)}$ of the
many Boson system
(see \cite[(\eqnTHuvoutput), (\eqnTHsmallfieldoutput) and (\eqnHTstartingpoint)]{PAR1}).
The approximations $\bbbt_{n}^{(SF)}$ to the block spin transformations
sketched above lead, for each $0\le n<n_p$, to the approximation of
\begin{equation}\label{eqnALFjn}
\begin{split}
J_n &= \int \Big[\prod_{x\in\cX_0^{(n)}} 
                   \sfrac{d\psi(x)^*\wedge d\psi(x)}{2\pi i}\Big]
    \Big( (\bbbs \bbbt_{n-1}^{(SF)}) \circ(\bbbs \bbbt_{n-2}^{(SF)})
    \circ\cdots \circ (\bbbs \bbbt_0^{(SF)}) \Big)\Big(e^{\cA_0} \Big)
    (\psi^*,\psi)\chi_n(\psi) \\
   &= \sfrac{\tilde \cZ_n}{\tilde \cZ_{n+1}}
     \int \Big[\prod_{x\in\cX_0^{(n+1)}} 
                   \sfrac{d\psi(x)^*\wedge d\psi(x)}{2\pi i}\Big]\ 
     (\bbbs \bbbt_n)\ \Big(\big[(\bbbs \bbbt_{n-1}^{(SF)}) 
    \circ\cdots \circ (\bbbs \bbbt_0^{(SF)})\, e^{\cA_0}  \big]
    \chi_n\Big) 
\end{split}
\end{equation}
(see \cite[Remarks \remHTpropblockspin.i and 
\remHTbasicremarkonscaling.i]{PAR1} and \eqref{eqnSFCIn}, below) 
by a constant times
\begin{align*}
J_{n+1} = \int \Big[\prod_{x\in\cX_0^{(n+1)}} 
                   \sfrac{d\psi(x)^*\wedge d\psi(x)}{2\pi i}\Big]\ 
    \Big( (\bbbs \bbbt_n^{(SF)}) 
    \circ\cdots \circ (\bbbs \bbbt_0^{(SF)}) \Big)\Big(e^{\cA_0} \Big)
    (\psi^*,\psi) \chi_{n+1}(\psi)
\end{align*}
In  \cite{PAR1,PAR2} we did not say very much either about 
the ``small field'' cutoff functions $\chi_n(\psi)$ or about 
the errors introduced by these approximations. In this note 
we make a possible choice of $\chi_n(\psi)$,
$n\ge 1$ (one of many possible choices) and argue that it is reasonable to expect that, for all $n\ge 0$, 
the error
$
\sfrac{1}{\tilde\cZ_{n+1}}J_{n+1}-\sfrac{1}{\tilde\cZ_n}J_n
$ 
is nonperturbatively small. By this we mean smaller than the
dominant contribution by a factor of order $O(e^{-1/\fv_n^\veps})$ 
for some $\veps>0$. We concentrate on the case $n\ge 1$. The case $n=0$
is similar but simpler.

We use two mechanisms for ``generating nonperturbatively small factors''. 
The first consists in exhibiting large negative contributions to the 
leading part $-A_n$ of the representation
\begin{equation}\label{eqnALFjnint}
\begin{split}
&
\Big( (\bbbs \bbbt_{n-1}^{(SF)}) \circ(\bbbs \bbbt_{n-2}^{(SF)})\circ
   \cdots \circ (\bbbs \bbbt_0^{(SF)}) \Big)
\Big(e^{\cA_0} \Big)
\cr
&\hskip 2cm
= \sfrac{1}{\cZ_n}\exp\Big\{- A_n(\psi^*,\psi, \phi_*, \phi,\,\mu_n,\cV_n)
+\cR_n +\cE_n 
\Big\}\bigg|_{\phi_{(*)} = \phi_{(*)n}(\psi^*,\psi,\mu_n,\cV_n)}
\end{split}
\end{equation}
of \cite[Theorem \thmTHmaintheorem]{PAR1}. These large negative contributions arise
whenever $|\psi(x)|$ or $|\partial_\nu\psi(x)|$ are sufficiently large
for some $x\in\cX_0^{(n)}$, $0\le \nu\le 3$. See Proposition
\ref{propALFfirst},  below. The second mechanism appears in the course of the
stationary phase approximation of \cite[\S\sectINTstatPhase]{PAR1} when 
$(\psi_*,\psi)$  is too far from the critical point $\big(\psi_{*n}(\th^*,\th), 
\psi_n(\th^*,\th)\big)$. See Step 3, below.

The background fields $\phi_{(*)n}$ and the actions $-A_n+\cR_n+\cE_n$ 
are well--defined on the ``domain of analyticity''
\begin{align*}
\An(n) & = \set{\psi\in\cH_0^{(n)}}{
                 |\psi(x)|<\ka(n),\ |\partial_\nu\psi(x)|<\ka'(n)
                 \text{ for all }x\in\cX_0^{(n)},\ 0\le\nu\le 3} 
\end{align*}
On these domains  we have the following lower and upper bounds on the
real part of the dominant contribution, $A_n$, to the action.

\begin{proposition}\label{propALFfirst}
Let $\de>0$.
There are constants $\ga,\tilde\ga>0$, independent of $\de$, such that 
if $\fv_0$ is sufficiently small, depending on $\de$, 
\begin{align*}
&\ga\sum_{\nu=0}^3\int_{\cX_0^{(n)}}\hskip-10pt dx\ |\partial_\nu\psi(x)|^2
-(1+\de)\mu_n \int_{\cX_0^{(n)}}\hskip-10pt  dx\ |\psi(x)|^2
+\half(1-\de)\rv_n \int_{\cX_0^{(n)}}\hskip-10pt  dx\ |\psi(x)|^4\cr
&\hskip0.2in
\le\Re A_n(\psi^*,\psi, \phi_*, \phi,\,\mu_n,\cV_n) 
  \Big|_{\phi_{(*)} = \phi_{(*)n}(\psi^*,\psi,\mu_n,\cV_n)} \cr
&\hskip0.4in\le
\tilde\ga\sum_{\nu=0}^3\int_{\cX_0^{(n)}}\hskip-10pt dx\ |\partial_\nu\psi(x)|^2
-(1-\de)\mu_n \int_{\cX_0^{(n)}}\hskip-10pt  dx\ |\psi(x)|^2
+\half(1+\de)\rv_n \int_{\cX_0^{(n)}}\hskip-10pt  dx\ |\psi(x)|^4
\end{align*}
for all $1\le n\le n_p$ and $\psi\in\An(n)$.
Here
\begin{equation*}
\rv_n=\int_{\cX_n^3} du_2\, du_3\, du_4\ V_n(0,u_2,u_3,u_4)
\end{equation*}
is the ``coupling constant at scale $n$''.

\end{proposition}

\begin{proof}
We start by recalling, from \cite[(\eqnBGAbgeqnsB) and 
Theorem \HTthminvertibleoperators]{PAR1},
that the background fields $\phi_{(*)n}(\psi^*,\psi,\mu_n,\cV_n)$
are  the solutions of
\begin{equation}\label{eqnLFbge}
\begin{split}
D_n^*\phi_* &=Q_n^* \fQ_n\psi^*
               -Q_n^*\fQ_n Q_n\phi_*+\mu_n\phi_*
               -\cV_{n*}'(\phi_*,\phi,\phi_*) \\
D_n\phi &=Q_n^* \fQ_n\psi
               -Q_n^*\fQ_n Q_n\phi+\mu_n\phi
               -\cV_n'(\phi,\phi_*,\phi)
\end{split}
\end{equation}
For the rest of this proof, we'll write 
$\phi_{(*)}$ instead of $\phi_{(*)n}(\psi^*,\psi,\mu_n,\cV_n)$ and 
$A_n$ instead of $A_n(\psi^*,\psi,\phi_*,\phi,\mu_n,\cV_n)$ . 
Substituting \eqref{eqnLFbge} into the definition \cite[(\eqnHTAndef)]{PAR1} of $A_n$ gives
\begin{equation}\label{eqnLFmainAction}
\begin{split}
A_n
%
&=\<\psi^*,\fQ_n\big(\psi-Q_n\phi\big)\>_0 
     -\< \phi_*,\,\cV_n'(\phi,\phi_*,\phi)\>_n
    + \cV_n(\phi_*,\phi) \\
&= \<\psi^*,\fQ_n\big(\psi-Q_n\phi\big)\>_0 
     - \cV_n(\phi_*,\phi)
\end{split}
\end{equation}
since
$\ 
\< \phi_*,\,\cV_n'(\phi,\phi_*,\phi)\>_n
 = 2 \cV_n(\phi_*,\phi)
\ $.

By \cite[Proposition \propBGEphivepssoln.a and Remark \remBGEphivepssoln]{BGE},
\begin{equation}\label{eqnLFphieqpsiplus}
\phi= \Phi +  \phi_n^{(\ge 3)}(\psi^*,\psi)
=\Phi  - S_n(\mu_n)\cV_n'\big(\Phi,\Phi_*,\Phi\big) 
          + \phi_n^{(\ge 5)}(\psi^*,\psi)
\end{equation}
with
\begin{equation*}
\Phi_*=\Phi_*(\mu_n)=S_n(\mu_n)^* Q_n^* \fQ_n\psi^*\qquad
\Phi=\Phi(\mu_n)=S_n(\mu_n) Q_n^* \fQ_n\psi
\end{equation*}
being the parts of $\phi_{(*)}$ that are of degree precisely one in 
$\psi^{(*)}$ and $\phi_n^{(\ge d)}(\psi^*,\psi)$ being the part that
is of degree at least $d$ in $\psi^{(*)}$. So
\begin{equation}\label{eqnLFpsiminusQphi}
\psi-Q_n\phi
=B^\De_n\psi
      + Q_n S_n(\mu_n)\cV_n'\big(\Phi,\Phi_*,\Phi\big) 
          -Q_n \phi_n^{(\ge 5)}(\psi^*,\psi)
\end{equation}
where $B^\De_n=\bbbone-Q_n S_n(\mu_n)Q_n^* \fQ_n$.
For general ($O(1)$ small enough) $\mu$
\begin{align*}
S_n(\mu) &= (D_n+Q_n^*\fQ_n Q_n-\mu)^{-1}
          = S_n(\bbbone-\mu S_n)^{-1}
         = S_n +\mu S_n S_n(\mu)
\end{align*}
so that, by \cite[Proposition \propHTexistencecriticalfields]{PAR1}, 
\begin{equation}\label{eqnLFBnDebis}
\begin{split}
\big<\psi^*,\fQ_n B^\De_n \psi\big>_0
&= \big<\psi^*,\,\fQ_n ( \bbbone-Q_n S_n Q_n^* \fQ_n) \psi\big>_0
- \mu_n \big<\psi^*,\,\fQ_n Q_n S_n(\mu_n)S_n Q_n^* \fQ_n ) \psi\big>_0
\\
&=  \big<\psi^*,\,\De^{(n)}\psi\big>_0  -\mu_n \<\Phi_*(\mu_n)\,,\,\Phi(0)\>_n
\end{split}
\end{equation}
with $\De^{(n)}=\De^{(n)}(\mu=0)$.
Inserting \eqref{eqnLFpsiminusQphi} and \eqref{eqnLFBnDebis} into the representation
\eqref{eqnLFmainAction} of $A_n$ gives
\begin{align*}
A_n&=  \big<\psi^*,\,\De^{(n)}\psi\big>_0
                  -\mu_n \<\Phi_*(\mu_n),\,\Phi(0)\>_n
 +  \<S_n(\mu_n)^*Q_n^*\fQ_n\psi^* , \cV_n'\big(\Phi,\Phi_*,\Phi\big)\>_n
         \\&\hskip6cm
 - \cV_n(\phi_*,\phi)
 - \big<Q_n^*\fQ_n\psi^*,  \phi_n^{(\ge 5)}(\psi^*,\psi)\big>_n\\
&=  \big<\psi^*,\,\De^{(n)}\psi\big>_0 
  -\mu_n \<\Phi_*(\mu_n)\,,\,\Phi(0)\>_n
  +  2 \cV_n\big(\Phi_*,\Phi\big)
  -  \cV_n\big(\Phi_*+\phi_{*n}^{(\ge 3)},\Phi+\phi_n^{(\ge 3)}\big)
\\&\hskip7.8cm
                  - \big<Q_n^*\fQ_n\psi^*,
                                   \phi_n^{(\ge 5)}(\psi^*,\psi)\big>_n
\\
&=  \big<\psi^*,\,\De^{(n)}\psi\big>_0 
  -\mu_n \<\Phi_*(\mu_n)\,,\,\Phi(0)\>_n
  +   \cV_n\big(\Phi_*,\Phi\big)
\\
&\hskip1cm
 -\Big[   \cV_n\big(\Phi_*+\phi_{*n}^{(\ge 3)},\Phi+\phi_n^{(\ge 3)}\big)
 - \cV_n\big(\Phi_*,\Phi\big)    \Big]
  - \big<Q_n^*\fQ_n\psi^*, \,\phi_n^{(\ge 5)}(\psi^*,\psi)\big>_n
\end{align*}
\smallskip

In our bounds, we fix $\psi\in\cH_0^{(n)}$ and denote
\begin{equation*}
\wf=\|\psi\|_{L^\infty}\quad
\wf_2=\|\psi\|_{L^2} \quad
\wf_4=\|\psi\|_{L^4} \quad
\wf'=\max_{0\le\nu\le 3}\|\partial_\nu\psi\|_{L^\infty}\quad
\wf'_2=\sum_{\nu=0}^3\|\partial_\nu\psi\|_{L^2}
\end{equation*}
Since $\psi\in\An(n)$, we have $\wf <\ka(n)$ and $\wf' < \ka'(n)$. Also,
 $\wf' \le 2\wf$ since $\partial_\nu$ is a difference operator on a unit
lattice.
By \cite[Remark \remBGEphivepssoln]{BGE} and \cite[Lemma \lemSUBLp.b]{SUB}, 
we have
$\, \big\| \phi_n^{(\ge 5)} \big\|_{L^{4/3}} = O\big(\fv_n^2 \wf^2\wf_4^3\big)
\,$
and consequently, by \cite[Lemma \lemSUBgenLoneLinfty]{SUB},
\begin{equation*}
\big<Q_n^*\fQ_n\psi^*, \,\phi_n^{(\ge 5)}(\psi^*,\psi)\big>_n
= O\big(\fv_n^2 \wf^2\wf_4^4\big)
\end{equation*}
Also, by \cite[Proposition \propBGEphivepssoln.a]{BGE}, \cite[Lemma \lemSUBdiffnorm,
 Remark \remSUBcB.a and Lemma \lemSUBLp.a]{SUB},
\begin{equation*}
\big|   \cV_n\big(\Phi_*+\phi_{*n}^{(\ge 3)},\Phi+\phi_n^{(\ge 3)}\big)
 - \cV_n\big(\Phi_*,\Phi\big)    \big| = O\big(\fv_n^2 \wf^2\wf_4^4\big)
\end{equation*}
Thus
\begin{equation}\label{eqnLFBsecondrepAn}
A_n=  \big<\psi^*,\,\De^{(n)}\psi\big>_0 
  -\mu_n \<\Phi_*(\mu_n)\,,\,\Phi(0)\>_n
  +   \cV_n\big(\Phi_*,\Phi\big) 
+O\big(\fv_n^2 \wf^2\wf_4^4\big)
\end{equation}

By \cite[Lemma \lemBGEexpandalphiveps]{BGE}
\begin{align*}
\Phi_{(*)}(\mu)(u)
=\big(S_n(\mu)^{(*)}Q_n^* \fQ_n\psi^{(*)}\big)(u)
&=\sfrac{a_n}{a_n-\mu}\Psi^{(*)}(u)
          +F_{\lb(*)}(\mu)(\{\partial_\nu\psi_{(*)}\})(u)
\end{align*}
with $\Psi^{(*)}(u)=\psi^{(*)}\big(X(u)\big)$ and with the maps $F_{\lb(*)}(\mu)$ being of degree precisely one.
Hence, recalling that $\wf'\le 2\wf$, 
\begin{equation}\label{eqnLFBPhipsi}
\begin{split}
\big|\mu_n \<\Phi_*(\mu_n),\,\Phi(0)\>_n 
            - \sfrac{a_n\mu_n}{a_n-\mu_n} \<\psi^*,\,\psi\>_0\big|
&=  |\mu_n|\, O\big( \wf_2\,\wf'_2+{\wf'}^2_{\!\!2}\big) 
= \sqrt{|\mu_n|}\,O\big( |\mu_n| \wf_2^2+{\wf'}^2_{\!\!2}\big) 
\\
\big|  \cV_n\big(\Phi_*,\Phi\big) 
- \big(\sfrac{a_n}{a_n-\mu_n}\big)^4  \cV_n\big(\Psi^*,\Psi\big)\big|
& =\fv_n\, O\big( {\wf'_2}\wf\wf_4^2
                  + {\wf'}^2_{\!\!2}\wf^2
                  + \wf'\wf\,{\wf'}^2_{\!\!2}+ {\wf'}^2{\wf'}^2_{\!\!2}\big)\\
   &= \sqrt{\fv_n}\,\wf \, 
          O\big( (1+\sqrt{\fv_n}\, \wf) {\wf'}^2_{\!\!2} +\fv_n\wf_4^4\big)
\end{split}
\end{equation}
Using \cite[Theorem \thmTHmaintheorem]{PAR1} and localizing as in 
\cite[Corollary \corLprelocalize]{PAR2},
\begin{equation}\label{eqnLFc}
\cV_n\big(\Psi^*,\Psi\big)
= \half\rv_n\wf_4^4 
+O\big(\sqrt{\fv_n\ka(n)^2}\,\big)\,{\wf'}^2_{\!\!2}
+O\big(\fv_0^{\sfrac{2}{3}-7\eps}+\sqrt{\fv_n\ka(n)^2}\,\big)\fv_n\wf_4^4
\end{equation}
Inserting \eqref{eqnLFBPhipsi} and \eqref{eqnLFc} into \eqref{eqnLFBsecondrepAn} we get
\begin{equation}\label{eqnLFa}
\begin{split}
 A_n&= \big<\psi^*\,,\,\De^{(n)}\psi\big>_0 
    -\sfrac{a_n\mu_n}{a_n-\mu_n} \<\psi^*\,,\,\psi\>_0
                   +\half \big(\sfrac{a_n}{a_n-\mu_n}\big)^4 \rv_n\wf_4^4 
                       \\
&\hskip1cm
     +O\big(\sqrt{\mu_n}+\sqrt{\fv_n \ka(n)^2}\,\big)\,{\wf'}^2_{\!\!2}
     + O\big(\sqrt{\mu_n}\big)\,\mu_n\wf_2^2
     +O\big(\sqrt{\fv_0}+\sqrt{\fv_n \ka(n)^2}\,\big)\,\fv_n\wf_4^4
\end{split}
\end{equation}

By \cite[Lemma \lemPOCDenppties.b,d]{POA}, the Fourier transform of 
$\De^{(n)}$ is
\begin{align*}
\widehat{\De^{(n)}}_k
&=-ik_0+\big(\sfrac{1}{a_n}+\sfrac{\veps_n^2}{2}\big)k_0^2
    +\half\! \smsum_{\nu,\nu'=1}^3 \!\!H_{\nu,\nu'}\bk_\nu\bk_{\nu'}
        +O\big(|k|^3\big)
\end{align*}
and obeys $\Re\widehat{\De^{(n)}(0)}_k\ge \rho(c)$ when $|k|\ge c$. In particular,
there are constants $\ga$, $\tilde\ga$, (independent of $n$ and $L$) such that
\begin{equation*}
8\ga(k_0^2+\bk^2)
\le \Re\widehat{\De^{(n)}}_k
\le \sfrac{1}{2}\tilde\ga(k_0^2+\bk^2)
\implies 2\ga{\wf'}^2_{\!\!2}
\le \Re \big<\psi^*\,,\,\De^{(n)}\psi\big>_0
 \le \half\tilde\ga{\wf'}^2_{\!\!2} 
\end{equation*}

It now suffices to combine \eqref{eqnLFa}--\eqref{eqnLFc} and use that,
by \cite[(\eqnPARestrad.a,b) and Corollary \corPARmunvn.a]{PAR1},
\begin{equation*}
\fv_n\ka(n)^2 <\fv_0^{\sfrac{3}{2}\eps}
\qquad
|\mu_n|< 4\fv_0^{5\eps}
\qquad
\half\le a_n\le 2
\end{equation*}
\end{proof}

We choose the ``small field'' cutoff function $\chi_n(\psi)$ of \eqref{eqnALFjn}
to be the characteristic function of
\begin{equation}\label{eqnALFIntncDef}
\begin{split}
\Int(n,\cc) & = \big\{\ \psi\in\cH_0^{(n)}\ \big|\ 
                 |\psi(x)|<\cc\ka(n),\ |\partial_\nu\psi(x)|<\cc\ka'(n)
                 \cr &\hskip2in
                 \quad\text{for all } x\in\cX_0^{(n)},\ 0\le\nu\le 3\ \big\} \cr
\end{split}
\end{equation}
with an appropriate value of $\cc$. 

Using the procedure starting at \cite[(\eqnHTmultiBS)]{PAR1} and leading up to 
\cite[Definition \defHTapproximateblockspintr]{PAR1}, and then applying
\cite[Theorem \thmTHmaintheorem]{PAR1}, we would expect the answer to the integral 
\begin{equation*}
\int_{\Int(n,\cc)} \Big[\prod_{x\in\cX_0^{(n)}} 
                   \sfrac{d\psi(x)^*\wedge d\psi(x)}{2\pi i}\Big]\ 
      e^{- A_n(\psi^*,\psi, \phi_*, \phi,\,\mu_n,\cV_n) +\cR_n +\cE_n }
     \bigg|_{\phi_{(*)} = \phi_{(*)n}(\psi^*,\psi,\mu_n,\cV_n)}
\end{equation*}
to have the main contribution a normalization constant times 
\begin{equation*}
\int_{\Int(n+1,\cc)} \Big[\prod_{x\in\cX_0^{(n+1)}} 
                   \sfrac{d\psi(x)^*\wedge d\psi(x)}{2\pi i}\Big]\ 
e^{- A_{n+1}(\psi^*,\psi, \phi_*, \phi,\,\mu_n,\cV_{n+1}) }
     \bigg|_{\phi_{(*)} = \phi_{(*)n+1}(\psi^*,\psi,\mu_{n+1},\cV_{n+1})}
\end{equation*}
The logarithm of the normalization constant is bounded in magnitude by
a constant, which depends only on $L$ and $\Gam_\op$, times $|\cX_0^{(n)}|$.
For constant $\psi$ close to the bottom of the potential well, the
integrand has magnitude greater than one, by the upper bound of
Proposition \ref{propALFfirst}.
Observe that if $\psi\in\An(n)\setminus \Int(n,\cc)$
then there is some $x\in\cX_0^{(n)}$ and possibly some $0\le\nu\le 3$ 
such that either $|\psi(x)|\ge \cc\ka(n)$ or 
$|\partial_\nu\psi(x)|\ge\cc\ka'(n)$.
So the lower bound of Proposition \ref{propALFfirst}, suggests the following
``corollary''. The significance of the quotation marks is that this is
a ``moral'' rather than a ``mathematical'' statement.

\begin{mcorollary}\label{corALFfirst}
Let $\cc>0$ and let $\fv_0$ be small enough, depending on $\cc$.
Then, for any $S\subset \An(n)\setminus \Int(n,\cc)$
\begin{equation*}
\int_S \Big[\prod_{x\in\cX_0^{(n)}} 
                   \sfrac{d\psi(x)^*\wedge d\psi(x)}{2\pi i}\Big]\ 
      \Big|e^{- A_n(\psi^*,\psi, \phi_*, \phi,\,\mu_n,\cV_n) +\cR_n +\cE_n }
      \Big|
\end{equation*}
is nonperturbatively small.
\end{mcorollary}

We shall later choose a small, possibly $L$--dependent constant,
$\cc_0>0$. Then our cutoff functions $\chi_n(\psi)$ are chosen to be
$\Int(n)=\Int(n,\cc_0)$.
 With these cutoff functions, we now sketch the argument that 
$
\sfrac{1}{\tilde\cZ_{n+1}}J_{n+1}-\sfrac{1}{\tilde\cZ_n}J_n
$ 
is nonperturbatively small in the case that $n\ge 1$.
It goes in three steps. First we just state what the steps are.
We'll discuss them in more detail shortly.

\noindent{\it Step 1:}\ \ \ 
Substituting
\begin{equation*}
1=\sfrac{1}{N^{(n)}_\bbbt}\hskip-2pt\int\hskip-2pt\Big[\hskip-5pt
 \prod_{y\in\cX_{-1}^{(n+1)}}\hskip-3pt
  \sfrac{d\th(y)^*\wedge d\th(y)}{2\pi i}\Big]
e^{-a L^{-2}\|\th-Q\psi\|^2_{-1}}
\end{equation*}
from \cite[Remark \remHTpropblockspin.a]{PAR1}, into \eqref{eqnALFjn} 
and \eqref{eqnALFjnint}, we have
\begin{equation}\label{eqnSFCIn}
\begin{split}
J_n &= \sfrac{1}{N^{(n)}_\bbbt \cZ_n}\hskip-2pt\int\hskip-2pt\Big[\hskip-5pt
 \prod_{y\in\cX_{-1}^{(n+1)}}\hskip-3pt
  \sfrac{d\th(y)^*\wedge d\th(y)}{2\pi i}\Big]
\int_{\Int(n)} \Big[\prod_{x\in\cX_0^{(n)}} 
                   \sfrac{d\psi(x)^*\wedge d\psi(x)}{2\pi i}\Big]\ 
     \\
& \hskip1.5in e^{-a L^{-2}\|\th-Q\psi\|^2_{-1}
    - A_n(\psi^*,\psi, \phi_*, \phi,\,\mu_n,\cV_n)
    +\cR_n+\cE_n}
\end{split}
\end{equation}
with
$
\phi_{(*)} = \phi_{(*)n}(\psi^*,\psi,\mu_n,\cV_n)
$.
The domain of integration for the double integral in \eqref{eqnSFCIn} is 
$(\th,\psi)\in \cH_{-1}^{(n+1)}\times\Int(n)$. 
The first step consists in  restricting the domain to
$(\th,\psi)\in \check\Int(n)\times\Int(n)$ where
\begin{align*}
&\check\Int(n)  = \set{\bbbs^{-1}\psi\in\cH_{-1}^{(n+1)}}{
                 \psi\in\Int(n+1)} \\
&\hskip0.1in = \set{\th\in\cH_{-1}^{(n+1)}}{
                 \!|\th(y)|<\cc_0\sfrac{\ka(n+1)}{L^{\sfrac{3}{2}}},\ |\partial_\nu\th(y)|<\cc_0\sfrac{\ka'(n+1)}{L^{\sfrac{3}{2}}L_\nu}
                 \ \forall\,y\in\cX_{-1}^{(n+1)},\ 0\le\nu\le 3} 
\end{align*}
and showing that the difference between $N^{(n)}_\bbbt \cZ_n J_n$ and 
\begin{equation}\label{eqnSFCJnB}
\begin{split}
&  \int_{\check\Int(n)}\hskip-2pt\Big[\hskip-5pt
 \prod_{y\in\cX_{-1}^{(n+1)}}\hskip-3pt
  \sfrac{d\th(y)^*\wedge d\th(y)}{2\pi i}\Big]
\int_{\Int(n)} \Big[\prod_{x\in\cX_0^{(n)}} 
                   \sfrac{d\psi(x)^*\wedge d\psi(x)}{2\pi i}\Big]\ 
     \\
& \hskip2in 
e^{-a L^{-2}\|\th-Q\psi\|^2_{-1}
    - A_n(\psi^*,\psi, \phi_*, \phi,\mu_n,\cV_n)
    +\cR_n+\cE_n}
\end{split}
\end{equation}
is nonperturbatively small.
\smallskip

\centerline{\includegraphics{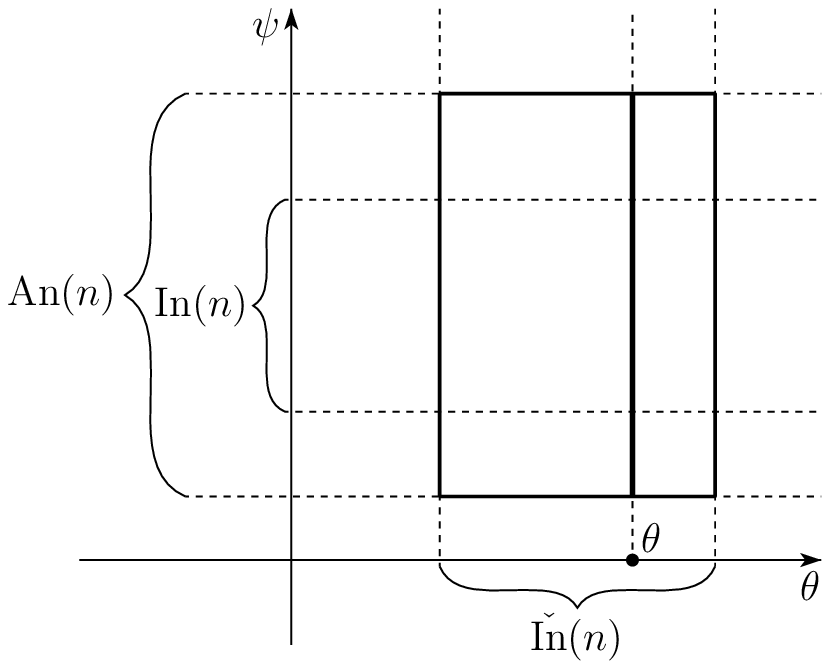}}

\smallskip
\noindent{\it Step 2:}\ \ \ This step consists in  enlarging the
integration domain $\check\Int(n)\times\Int(n)$ 
of \eqref{eqnSFCJnB} to $\check\Int(n)\times\An(n)$ and showing that the difference
between \eqref{eqnSFCJnB} and 
\begin{equation}\label{eqnSFCJnC}
\begin{split}
& \int_{\check\Int(n)}\hskip-2pt\Big[\hskip-5pt
 \prod_{y\in\cX_{-1}^{(n+1)}}\hskip-3pt
  \sfrac{d\th(y)^*\wedge d\th(y)}{2\pi i}\Big]
\int_{\An(n)} \Big[\prod_{x\in\cX_0^{(n)}} 
                   \sfrac{d\psi(x)^*\wedge d\psi(x)}{2\pi i}\Big]\ 
     \\
& \hskip2in 
e^{-a L^{-2}\|\th-Q\psi\|^2_{-1}
    - A_n(\psi^*,\psi, \phi_*, \phi, \mu_n, \cV_n)
    +\cR_n+\cE_n}
\end{split}
\end{equation}
is nonperturbatively small.

\smallskip
\noindent{\it Step 3:}\ \ \ The third step consists in showing that,
for each fixed $\th \in \check\Int(n)$ the inner integral 
\begin{equation}\label{eqnLFinnerInt}
\begin{split}
\int_{\An(n)} \Big[\prod_{x\in\cX_0^{(n)}} 
                   \sfrac{d\psi(x)^*\wedge d\psi(x)}{2\pi i}\Big]\ 
e^{-a L^{-2}\|\th-Q\psi\|^2_{-1}
    - A_n(\psi^*,\psi, \phi_*, \phi, \mu_n, \cV_n)
    +\cR_n+\cE_n}
\end{split}
\end{equation}
of \eqref{eqnSFCJnC} is nonperturbatively close to 
$\det C^{(n)}\ e^{\check\cC_n(\th_*,\th)}\ \check\cF_n(\th_*,\th)$
with the $\check\cC_n(\th_*,\th)$ and $\check\cF_n(\th_*,\th)$
of \cite[Proposition \propSTmainProp.a]{PAR1}. (They are defined at the beginning of
of \cite[\S\chapSTstrategy]{PAR1}.) 

Putting these three steps together, we see that $\sfrac{1}{\tilde\cZ_n}J_n$
is nonperturbatively close to
\begin{align*}
&\sfrac{\det C^{(n)}}{N^{(n)}_\bbbt \cZ_n \tilde\cZ_n }
 \int_{\check\Int(n)}\hskip-2pt\Big[\hskip-5pt
 \prod_{y\in\cX_{-1}^{(n+1)}}\hskip-3pt
  \sfrac{d\th(y)^*\wedge d\th(y)}{2\pi i}\Big]\
  e^{\check\cC_n(\th_*,\th)}\ \check\cF_n(\th_*,\th) \\
&= \sfrac{1}{\tilde\cZ_n }
 \int_{\check\Int(n)}\hskip-2pt\Big[\hskip-5pt
 \prod_{y\in\cX_{-1}^{(n+1)}}\hskip-3pt
  \sfrac{d\th(y)^*\wedge d\th(y)}{2\pi i}\Big]\ \Big( 
  \bbbt_n^{(SF)}\circ(\bbbs \bbbt_{n-1}^{(SF)}) \circ
   \cdots \circ (\bbbs \bbbt_0^{(SF)}) \Big)
\Big(e^{\cA_0} \Big)(\th^*,\th) \\
&=\sfrac{1}{\tilde\cZ_n  L^{3|\cX_0^{(n+1)}|}}
   \int_{\Int(n+1)} \Big[\prod_{x\in\cX_0^{(n+1)}} 
                   \sfrac{d\psi(x)^*\wedge d\psi(x)}{2\pi i}\Big]\ 
    \Big( (\bbbs \bbbt_n^{(SF)}) 
    \circ\cdots \circ (\bbbs \bbbt_0^{(SF)}) \Big)\Big(e^{\cA_0} \Big)
    (\psi^*,\psi) \\
&=\sfrac{1}{\tilde\cZ_{n+1}}J_{n+1}
\end{align*}
by  \eqref{eqnALFjnint}, \cite[Definition \defHTapproximateblockspintr,
Proposition \propSTmainProp.a and (\eqnHTmultiBS)]{PAR1}.

\medskip
We now elaborate on these three steps.

\noindent{\it Step 1:}\ \ \ 
Fix any 
$\th\notin \check\Int(n)$ and decompose the domain of integration 
for the $\psi$ integral
\begin{align*}
\Int(n)= \Int_s(n,\th)\cup\Int_b(n,\th)
\end{align*}
with
\begin{align*}
\Int_s(n,\th)
   &=\set{\psi\in\Int(n)}{L^{-1}\|\th-Q\psi\|_{-1}<\sfrac{1}{\rv_n^\eps}} \\
\Int_b(n,\th)
   &=\set{\psi\in\Int(n)}{L^{-1}\|\th-Q\psi\|_{-1}\ge \sfrac{1}{\rv_n^\eps}} 
\end{align*}
We would expect that the integral over $\psi\in \Int_b(n,\th)$ gives a 
nonperturbatively small contribution
because of the $-a L^{-2}\|\th-Q\psi\|^2_{-1}$ in the exponent.
Furthermore we claim that $\Int_s(n,\th)\subset \An(n)\setminus \Int(n,\cc)$
with 
\begin{equation*}
\cc=\min_{0\le\nu\le 3}
           \Big\{ \sfrac{\cc_0}{2 L^{\frac{3}{2}-\eta}\|Q\|_{m=0}}\,,\,
            \sfrac{\cc_0}{2 L^{\frac{3}{2}-\eta'}L_\nu\|Q^{(-)}_{+,\nu}\|_{m=0}}
   \Big\}
\end{equation*}
(The operator $Q_{n,\nu}^{(-)}$ was defined in \cite[(\eqnPBSqnplusminus)]{POA},
$L_0=L^2$ and $L_\nu=L$ for $\nu=1,2,3$.)
This will ``imply'', by ``Corollary \ref{corALFfirst}'', that the integral 
over $\Int_s(n,\th)$ is also nonperturbatively small.
So let $\psi\in\Int_s(n,\th)$.
\begin{itemize}[leftmargin=*, topsep=2pt, itemsep=2pt, parsep=0pt]
\item
If there is a $y\in\cX_{-1}^{(n+1)}$ with
$|\th(y)|\ge\cc_0\sfrac{\ka(n+1)}{L^{\frac{3}{2}}}$, then since
\begin{align*}
|\th(y)| & \le |\th(y)-(Q\psi)(y)| + |(Q\psi)(y)| \cr
  &\le \sfrac{1}{L^{5/2}} \|\th-Q\psi\|_{-1} + \|Q\|_{m=0}\|\psi\|_{\ell^\infty}\\
  & \le  \sfrac{1}{L^{3/2}\rv_n^\eps}  + \|Q\|_{m=0}\|\psi\|_{\ell^\infty}
\end{align*}
we have
\begin{equation*}
\|\psi\|_{\ell^\infty} 
   \ge \sfrac{1}{L^{3/2}\|Q\|_{m=0}}
               \big(\cc_0\ka(n+1)  - \sfrac{1}{\rv_n^\eps}\big)
   \ge \sfrac{\cc_0}{2 L^{3/2}\|Q\|_{m=0}} \ka(n+1)
   \ge \cc \ka(n)
\end{equation*}
and $\psi\notin \Int(n,\cc)$.

\item
If there is a $y\in\cX_{-1}^{(n+1)}$ and a $0\le\nu\le 3$
with $|\partial_\nu\th(y)|\ge\cc_0\sfrac{\ka'(n+1)}{L^{3/2}L_\nu}$, 
then since
\begin{align*}
|\partial_\nu\th(y)| & \le |\partial_\nu\th(y)-(\partial_\nu Q\psi)(y)| 
       + |(\partial_\nu Q\psi)(y)| \\
  & = \big|\partial_\nu\big[\th- Q\psi\big](y)\big| 
       + |( Q^{(-)}_{+,\nu}\partial_\nu \psi)(y)| \\
  &\le \sfrac{2}{L^{5/2}L_\nu} \|\th-Q\psi\|_{-1} 
        + \|Q^{(-)}_{+,\nu}\|_{m=0}\|\partial_\nu\psi\|_{\ell^\infty}\\
  & \le  \sfrac{2}{L^{3/2}L_\nu\rv_n^\eps}  
       + \|Q^{(-)}_{+,\nu}\|_{m=0}\|\partial_\nu\psi\|_{\ell^\infty}
\end{align*}
we have
\begin{align*}
\|\partial_\nu\psi\|_{\ell^\infty} 
   &\ge \sfrac{1}{L^{3/2}L_\nu\|Q^{(-)}_{+,\nu}\|_{m=0}}
               \big(\cc_0\ka'(n+1)  - \sfrac{2}{\rv_n^\eps}\big)
  \ge \sfrac{\cc_0}{2 L^{3/2}L_\nu\|Q^{(-)}_{+,\nu}\|_{m=0}} \ka'(n+1)
  \ge \cc \ka'(n)
\end{align*}
and, again, $\psi\notin \Int(n,\cc)$.
\end{itemize}

\vskip0.1in
\noindent{\it Step 2}\ \ \ ``follows'' directly from ``Corollary \ref{corALFfirst}''\ 
with $S=\An(n)\setminus \Int(n)$.

\vskip0.1in
\noindent{\it Step 3:}\ \ \ 
Fix any $\th\in\check\Int(n)$. Set
$$
\rho_n(\th)
=\psi_{*n}(\th^*,\th)^*-\psi_n(\th^*,\th)
$$
By Remark \cite[\remBGEcrfandcomplexconj]{BGE},
\begin{equation*}
\|\rho_n(\th)\|_\infty \le\sqrt{\cc_0}\,\ka'(n)
\end{equation*}
if $\cc_0$ is small enough.
Clearly $(\th^*,\th)$ is in the
domain of $\psi_{(*)n}$. Recall that
\begin{equation*}
\psi_{(*)n}(\th_*,\th)
=\sfrac{1}{L^{3/2}}\bbbl_*\big[\hat \psi_{(*)n}(\bbbs\th_*,\bbbs\th)\big]
\end{equation*}
Since  $\|\th\|_\infty<\sfrac{\cc_0}{L^{3/2}}\ka(n+1)$ and $\|\partial_\nu\th\|_\infty<\sfrac{\cc_0}{L^{3/2}L_\nu}\ka'(n+1)$
for all $0\le\nu\le 3$, we have, by \cite[Proposition \propCFpsisoln]{BGE}, 
with $\wf=\cc_0\bar\ka$, $\wf'=\cc_0\bar\ka'$,
\begin{alignat*}{3}
\|\psi_{(*)n}\|_\infty
&=\sfrac{1}{L^{3/2}}\|\hat\psi_{(*)n}\|_\infty&
&\le\sfrac{\cc_0}{L^{3/2}} \GGa_\op \ka(n+1)&
&=\sfrac{\cc_0}{L^{3/2-\eta}} \GGa_\op \ka(n) \\
\|\partial_\nu\psi_{(*)n}\|_\infty
&=\sfrac{1}{L^{3/2}L_\nu}\|\partial_\nu\hat\psi_{(*)n}\|_\infty&
&\le\sfrac{\cc_0}{L^{3/2}L_\nu} \GGa_\op \ka'(n+1)&
&=\sfrac{\cc_0}{L^{3/2-\eta'}L_\nu} \GGa_\op \ka'(n) 
\end{alignat*}
with $\eta<\sfrac{7}{8}$ and $\eta'<\sfrac{3}{4}$. 
Hence, if we pick $L$ large enough or $\cc_0$ small enough, 
depending only on $\GGa_\op$,
\begin{equation*}
\psi_{(*)n}(\th^*,\th) + \de\psi_{(*)} \in\An(n)\qquad
\end{equation*}
for all $\de\psi$ obeying $\|\de\psi_{(*)}\|<\sfrac{1}{2}\ka(n)$, 
               $\|\partial_\nu \de\psi_{(*)}\|<\sfrac{1}{2}\ka'(n)$.
We may rewrite the integral \eqref{eqnLFinnerInt} as
\begin{align*}
&\int_{\An(n)} \Big[\prod_{x\in\cX_0^{(n)}} 
                   \sfrac{d\psi(x)^*\wedge d\psi(x)}{2\pi i}\Big]\ 
     e^{-a L^{-2}\|\th-Q\psi\|^2_{-1}
    - A_n(\psi^*,\psi, \phi_*, \phi, \mu_n, \cV_n)
    +\cR_n+\cE_n} 
=\int_{I_n(\th^*,\th)} \tilde\om_n
\end{align*}
where $\,\tilde\om_n\,$ is the holomorphic differential form obtained from
the integrand on the left hand side through the substitution
\begin{equation}\label{eqnSFCdepsisub}
\psi^*=\psi_{*n}(\th^*,\th)+\de\psi_*\qquad
\psi=\psi_n(\th^*,\th)+\de\psi\qquad
\end{equation}
and the domain
\begin{align*}
I_n(\th^*,\th)
=\Big\{\,(\de\psi_*,\de\psi) \in \cH_0^{(n)} \times \cH_0^{(n)} \,\Big|\, 
  \ \ &\de\psi=\de\psi_*^* + \rho_n(\th),\cr
  &\psi_{n}(\th^*,\th)+\de\psi\in\An(n)\,\Big\}
\end{align*}
As in the proof of \cite[Proposition \propSTmainProp]{PAR1},
\begin{align*}
\tilde\om_n & = e^{\check\cC_n(\th_*,\th)} 
               e^{-<\de\psi_*,{C^{(n)}}^{-1}\,\de\psi>_0}
               e^{-\de\check A_n(\th_*,\th,\de\psi_*,\de\psi)
                 +\de\check\cR_n+\de\check\cE_n}
               \prod_{x\in\cX_0^{(n)}} 
                   \sfrac{d\de\psi_*(x)\wedge d\de\psi(x)}{2\pi i}
\end{align*}
We next make the change of variables 
$\de\psi_*= D^{(n)*}\ze_*$, 
       $\de\psi= D^{(n)}\ze$, 
with $D^{(n)}$ being an operator square root of $C^{(n)}$ and
$D^{(n)*}$ being the transpose (not adjoint) of $D^{(n)}$.
Then
\begin{align*}
\int_{I_n(\th^*,\th)} \tilde\om_n
=\int_{I'_n(\th)} \om_n
\end{align*}
with
\begin{align*}
I'_n(\th)
=\Big\{\,(\ze_*,\ze) \in \cH_0^{(n)} \times \cH_0^{(n)} \,\Big|\, 
  \ \ &D^{(n)}\ze=\overline{D^{(n)*}}\ze_*^* + \rho_n(\th),\cr
  &\psi_n(\th^*,\th)+D^{(n)}\ze\in\An(n)\,\Big\}
\end{align*}
and
\begin{align*}
\om_n & =   \det C^{(n)}\   e^{\check\cC_n(\th_*,\th)} 
               e^{-<\ze_*,\ze>_0}
               e^{-\de\check A_n
                 +\de\check\cR_n+\de\check\cE_n}
               \prod_{x\in\cX_0^{(n)}} 
                   \sfrac{d\ze_*(x)\wedge d\ze(x)}{2\pi i}
\end{align*}
When $D^{(n)}\ze=\overline{D^{(n)*}}\ze_*^* + \rho_n$
we have 
$
\ze_*= {D^{(n)*}}^{-1}\overline{D^{(n)}}\ze^*-{{D^{(n)*}}}^{-1}\rho_n^*
$
so that
\begin{equation}\label{eqnSFCbilinearform}
\<\ze_*,\ze\>
= \big<{D^{(n)*}}^{-1}\overline{D^{(n)}}\ze^*,\ze\big>
          - \big<{{D^{(n)*}}}^{-1}\rho_n^*,\ze\big>
\end{equation}

To convert the integral $\int_{I'_n(\th)} \om_n$ into an integral of
$\om_n$ over the ``real'' disk
$$
\Bot=\set{(\ze_*,\ze)}{\ze_* =\ze^*,\ \|\ze\|<r_n }
$$
we now choose a ``Stokes' Cylinder'' $\cY$ that contains $\Bot$ in its
boundary.
For each $0\le t\le 1$, set $C(t)=\big[t{C^{(n)}}^{-1}+(1-t)\bbbone\big]^{-1}$.
Note that $C(t)^{-1}=\sfrac{at}{L^2}Q^*Q+t\De^{(n)}+(1-t)\bbbone$ 
has strictly positive real part. (See \cite[(\eqnPOCftCinverse), 
Lemma \lemPOCDenppties.b,d and Lemma \lemPBSuplusppties.c]{POA}.)
Denote by $D(t)$ the square root of $C(t)$ given by the contour integral
as in \cite[Corollary \corPOCsquareroot]{POA}. Set
\begin{align*}
I'_t&=\Big\{\,(\ze_*,\ze) \in \cH_0^{(n)} \times \cH_0^{(n)} \,\Big|\, 
   D(t)\ze= D(t)^\dagger\ze_*^* + t\rho_n(\th)\ \Big\}
\\
\cY&=\Big\{\,(\ze_*,\ze) \in \bigcup_{0\le t\le 1} I'_t\,\Big|\, 
    \ \|\ze\|\le \root{4}\of{\cc_0} \ka'(n)\Big\}
\end{align*}
If $(\ze_*,\ze)\in\cY$, then both $\psi_n(\th^*,\th)+D(t)\ze$ 
and $\psi_{*n}(\th^*,\th)+D(t)^*\ze_*$ are in $\An(n)$, by
\cite[Remark \remBGEcrfandcomplexconj]{BGE}, provided we choose $c_0$ small enough. 
This is illustrated, for $t=1$, in the figure

\centerline{\includegraphics{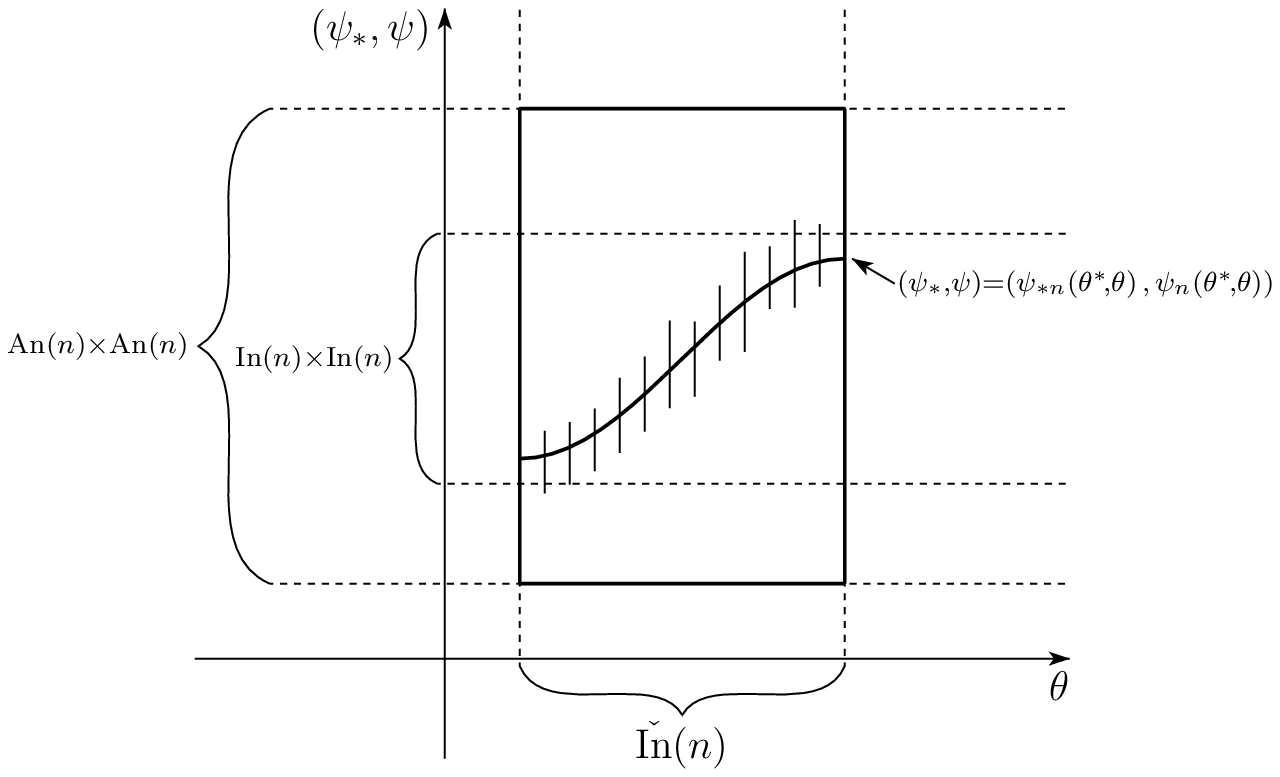}}
\noindent {$\sst\rm Each\ small\ vertical\ line\ in\ this\ figure\ is$}
\vskip-0.3in
\begin{align*}
\sst\set{\big(\psi_{*n}(\th^*,\th)+D^{(n)}\ze\,,\,
             \psi_{*n}(\th^*,\th)+{D^{(n)}}^*\ze\big)}
        {D^{(n)}\ze=\overline{D^{(n)*}}\ze_*^* + \rho_n(\th) \ 
         \|\ze\|_\infty\le \root{4}\of{\cc_0}\, \ka'(n)}
\end{align*}
\pagebreak[1]

\bigskip
We now show that if $(\ze_*,\ze)\in I'_t$, for some $0\le t\le 1$, 
and $\|\ze\|\ge \root{4}\of{\cc_0}\, \ka'(n)$, then 
\begin{equation}\label{eqnSFCbilinearformt}
\Re \<\ze_*,\ze\> \ge\const \|\ze\|^2\ge\const\sqrt{\cc_0}\,\ka'(n)^2
\gg r_n^2
\end{equation}
Since $(\ze_*,\ze)\in I'_t$
\begin{align*}
\<\ze_*,\ze\>
&= \big<{D(t)^*}^{-1}\overline{D(t)}\ze^*,\ze\big>
          - t\big<{D(t)^*}^{-1}\rho_n^*,\ze\big>
\end{align*}
The real part of $\big<{D(t)^*}^{-1}\overline{D(t)}\ze^*,\ze\big>$ is
\begin{align}
&\half\Big\{\big<{D(t)^*}^{-1}\overline{D(t)}\ze^*,\ze\big>
          +\big<\overline{{D(t)^*}^{-1}} D(t)\ze,\ze^*\big>\Big\}\nonumber\\
&\hskip0.5in=\half\big<\overline{D(t)}\ze^*,
   {D(t)^\dagger}^{-1}\big[D(t)^\dagger D(t)^{-1} +  
      {D(t)^\dagger}^{-1}D(t)\big]D(t)^{-1}D(t)\ze\big>\nonumber\\
&\hskip0.5in=\half\big<\eta^*,
     \big[D(t)^{-2}+(D(t)^{-2})^\dagger\big]\eta\big>
 \qquad\text{with $\eta=D(t)\ze$}\nonumber\\
&\hskip0.5in=\big<\eta^*,
     \big[\sfrac{at}{L^2}Q^*Q+t\Re \De^{(n)}+(1-t)\bbbone\big]
    \,\eta\big>\cr
&\hskip0.5in\ge \const\|\eta\|^2\nonumber\\
&\hskip0.5in\ge \const\|\ze\|^2  \label{eqnSFCrealPartQfT}
\end{align}
since $D(t)^{-2}=C(t)^{-1}=\sfrac{at}{L^2}Q^*Q+t\De^{(n)}+(1-t)\bbbone$
and since $D(t)^{-1}$ is a bounded operator.

Since  $\int_{\partial\cY} \om_n=0$ by Stokes' theorem,
\begin{align*}
\int_{I'_n(\th)} \om_n
&=\int_{I'_n(\th)\setminus\partial\cY} \om_n
  +\int_{I'_n(\th)\cap\partial\cY} \om_n
  -\int_{\partial\cY} \om_n
\cr
&=\int_{I'_n(\th)\setminus\partial\cY} \om_n
  -\int_{\partial\cY\setminus(I'_n(\th)\cap\partial\cY)} \om_n
\cr
&=\int_{\Bot} \om_n
    + \int_{I'_n(\th)\setminus\partial\cY} \om_n
    - \int_{\partial\cY\setminus
             [(I'_n(\th)\cap\partial\cY)\cup\Bot]} \om_n
\end{align*}

\centerline{\includegraphics{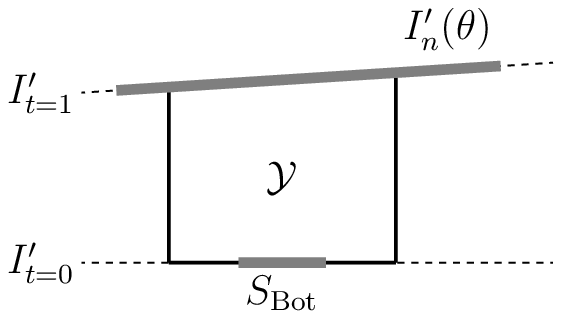}}

\noindent
By \eqref{eqnSFCbilinearformt} and the definition of $\Bot$,
the last two integrals are both nonperturbatively small.
By the definition of $\check\cF_n(\th_*,\th)$, the first integral 
\begin{equation*}
\int_{\Bot} \om_n=  \det C^{(n)}\ e^{\check\cC_n(\th_*,\th)}\ \check\cF_n(\th_*,\th)
\end{equation*}

\newpage
\bibliographystyle{plain}
\bibliography{refs}

\end{document}